\documentclass[a4paper, 10pt, twocolumn, twoside]{IEEEtran}
\IEEEoverridecommandlockouts
\usepackage{amsmath}
\usepackage{graphicx}
\usepackage{amssymb}
\usepackage{array}
\usepackage{enumerate}
\usepackage{bm}

\def\bfE{\mbox{\boldmath$E$}}

\newtheorem{theorem}{Theorem}
\newtheorem{lemma}{Lemma}

\newtheorem{corollary}{Corrollary}

\newtheorem{remark}{Remark}

\title{When MIMO Control Meets MIMO Communication: A Majorization Condition for Networked Stabilizability}
\author{Wei Chen, \emph{Student Member, IEEE}, Songbai Wang, \emph{Student Member, IEEE}, \\and Li Qiu, \emph{Fellow, IEEE}\vspace{-15pt}
\thanks{The authors are with the Department of Electronic and Computer Engineering, The Hong Kong University of Science and Technology,
Clear Water Bay, Kowloon, Hong Kong, China {\tt wchenust@ust.hk, swangas@ust.hk, eeqiu@ust.hk}}
}

\begin{document}

\maketitle
\thispagestyle{empty}
\pagestyle{empty}


\begin{abstract}
In this paper, we initiate the study of networked stabilization via a MIMO communication scheme between the controller and the plant.
Specifically, the communication system is modeled as a MIMO transceiver, which consists of three parts: an encoder, a MIMO channel, and a decoder. In the spirit of MIMO communication, the number of SISO subchannels in the transceiver is often greater than the number of data streams to be transmitted. Moreover, the subchannel capacities are assumed to be fixed a priori. In this case, the encoder/decoder pair gives an additional design freedom on top of the controller, leading to a stabilization problem via coding/control co-design. It turns out that how to take the best advantage of the coding mechanism is quite crucial. From a demand/supply perspective, the design of the coding mechanism boils down to reshaping the demands for communication resource from different control inputs to match the given supplies. We study the problem for the case of AWGN subchannels and fading subchannels, respectively. In both cases, we arrive at a unified necessary and sufficient condition on the capacities of the subchannels under which the coding/control co-design problem is solvable. The condition is given in terms of a majorization type relation. As we go along, systematic procedures are also put forward to implement the coding/control co-design. A numerical example is presented to illustrate our results.
\end{abstract}

\begin{keywords}
Networked stabilization, MIMO communication, coding/control co-design, majorization, topological entropy.
\end{keywords}

\section{Introduction}
The phrase ``MIMO systems'' has different meanings in control theory and in communication theory. As is well known, a MIMO control system refers to a multi-input multi-output physical system interconnected with a multi-input multi-output controller, while a MIMO communication system refers to a communication structure deployed to break the capacity limit of the conventional SISO communication scheme. What will happen when MIMO control meets MIMO communication? In this paper, we attempt to unveil the magic
by investigating
a particular networked stabilization problem, in which MIMO communication is utilized to transmit the control signals.

Generally speaking, a networked control system (NCS) is a feedback system wherein the feedback loop is closed over a communication network. For a better understanding of the background, we briefly review the state of the art as below. It has been well recognized that in networked control, whether stabilization can be achieved or not critically depends on the information constraints in the communication network. As such, a primary concern of networked stabilization is to find a fundamental limitation on the information constraints so as to render stabilization possible. For a single-input system, the networked stabilization problem has been extensively studied under different information constraints. See \cite{Baillieul,ne,Tatikonda} for data rate constraint, \cite{Braslavsky} for signal-to-noise ratio constraint, \cite{Eliaa,Fu} for quantization, \cite{Eliab} for fading effect, and \cite{Tsumura} for quantization and packet drop, etc. All these studies converge to a unified fundamental limitation on the information constraints required for stabilization given by the topological entropy of the open-loop plant, i.e., the logarithm of the absolute product of unstable poles for a discrete-time plant, or the sum of the unstable poles for a continuous-time plant.

The story gets more complicated when it comes to multi-input systems. In many existing works, e.g., \cite{Fu,Gao,Garone,Vargas}, a mere controller design problem is formulated assuming that the communication network is given a priori.
It turns out that research problems formulated in this way are usually very hard to solve. To mitigate this difficulty, the idea of channel resource allocation is proposed in \cite{Qiu} and followed by several other works such as \cite{Chenb,Feng,Xiaob,Xiaoa}, etc. Specifically, it is assumed therein that the channel capacities can be freely allocated among different input channels subject to a total capacity constraint. This in turn results in a channel/controller co-design problem that is shown to be solvable, if and only if the total channel capacity is greater than the topological entropy of the open-loop plant. Similar ideas, although not stated explicitly, can also be seen in \cite{Li,Shu}, in which the networked stabilization over parallel Gaussian channels are considered. There, only a total transmission power constraint is imposed while the power in each individual channel remains flexible.

Note that the vast majority of existing works assume a SISO communication scheme between the controller and the plant, i.e., each control input is transmitted through a dedicated SISO channel. One strong motivation of this work is drawn from the MIMO technology recently developed in communication theory. It has been widely used in wireless communication where spacial diversity can be exploited to increase the data transmission capacity \cite{TseV}. Inspired by the huge success of MIMO communication, we are eager to explore the potential merits brought about by
utilizing MIMO communication in MIMO control.

As a starting point, we investigate the stabilization of an NCS via MIMO communication. Specifically, we model the communication system between the controller and the plant as a MIMO transceiver, which consists of an encoder, a MIMO channel, and a decoder. One essence of MIMO communication is that the number of SISO subchannels in the transceiver is often greater than the number of data streams. When applied to networked control, this means that the number of subchannels is greater than the number of control inputs. In this paper,
the subchannels are modeled in two different ways. We first consider the AWGN subchannels and then the fading subchannels. In both cases, we assume that the subchannel capacities are fixed a priori and, thus, cannot be allocated as in \cite{Qiu,Chenb,Feng,Xiaob,Xiaoa}. Nevertheless, the encoder/decoder pair in the MIMO transceiver now serves as a substituted design freedom. The main message we convey through this paper is that to stabilize the feedback system, one should not only make the most of the controller but also design the coding mechanism wisely. This results in a stabilization problem via coding/control co-design. Quite surprisingly, a necessary and sufficient condition is obtained for the solvability of the coding/control co-design problem given explicitly in terms of a majorization type relation.

Note that majorization is an old yet new mathematical tool \cite{MOA}. It is old since it has been studied in mathematics by giants like Hardy, Littlewood, and P\'{o}lya in their masterpiece \cite{HLP} and has been widely used in statistics in the past 100 years. It is new since its applications in engineering only appear recently, notably in wireless communication \cite{PJ}, information theory \cite{CV}, operations research \cite{CY}, and power systems \cite{Nayyar}, etc. The application of majorization in control theory remains quite scattered in the literature.
One relevant work can be seen in \cite{Lipsa}, in which majorization is utilized to investigate the remote state estimation with communication costs for a first-order linear time-invariant system.

The rest of this paper is organized as follows.
Section II formulates the problem to be studied. Section III gives some preliminary
knowledge.
Section IV establishes the condition for networked stabilizability via MIMO communication over AWGN subchannels. The result is then extended to the case of fading subchannels in Section V. A numerical example is worked out in Section VI.
Finally, some conclusion remarks follow in Section VII.
Most notations in this paper are more or less standard and will be made clear as we proceed.

\section{Problem Formulation}
Consider the NCS depicted in Fig. \ref{feedback}. Here, the plant $[A|B]$ is a continuous-time linear time-invariant system described by the state space model
\begin{align*}
\dot{x}(t)=Ax(t)+Bu(t),\quad x(0)=x_0,
\end{align*}
where $A\in\mathbb{R}^{n\times n}$ and $B\in\mathbb{R}^{n\times m}$. Assume that $[A|B]$ is unstable but stabilizable. Let the state $x(t)$ be available for feedback. If the communication network between the controller and the plant is ideal, i.e., $u(t)=v(t)$, one can easily design a state feedback controller $v(t)=Fx(t)$ so that the closed-loop system is stable. However, such state feedback design faces challenges when the communication network is not ideal, i.e., $u(t)$ is only a distorted version of $v(t)$. In this case, the achievability of stabilization will depend on the transmission accuracy of the communication network. In fact, a general concern of networked stabilization is to find a fundamental limitation on the quality of the communication network so as to render stabilization possible.

\begin{figure}[htbp]
\begin{center}
\begin{picture}(70,16)
\thicklines
\put(0,5){\framebox(10,10){$F$}} \put(10,10){\vector(1,0){10}}
\put(35,10){\oval(30,10)}
\put(25,7.5){\makebox(20,10){Communication}}
\put(25,2.5){\makebox(20,10){Network}}
\put(50,10){\vector(1,0){10}}
\put(60,5){\framebox(10,10){$[A|B]$}}
\put(65,5){\line(0,-1){5}} \put(65,0){\line(-1,0){60}}
\put(5,0){\vector(0,1){5}} \put(10,10){\makebox(10,5){$v$}}
\put(50,10){\makebox(10,5){$u$}}
\put(35,2){\makebox(0,0){$x$}}
\end{picture}
\end{center}
\vspace{-2pt}
\caption{State feedback via communication network.}\label{feedback}
\end{figure}
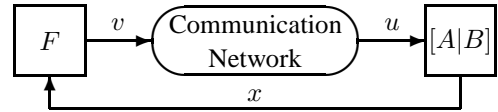

Notice that almost all existing studies, for instance, \cite{Chenb,Fu,Gao,Garone,Li,Qiu,Shu,Vargas,Xiaob,Xiaoa}, assume a SISO communication scheme between the controller and the plant, as shown in Fig. \ref{SISO}. Specifically, the communication network is modeled as a set of parallel independent SISO channels so that each control input is to be transmitted through a dedicated SISO channel. In spite of the multi-variable nature of the transmitted signal $v$ and the received signal $u$, the essence of such communication scheme is still SISO technology.
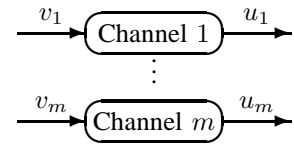
\begin{figure}[htbp]
\begin{center}
\begin{picture}(40,23)
\thicklines
\put(0,5){\vector(1,0){10}}
\put(20,5){\oval(20,6)}
\put(10,0){\makebox(20,10){Channel $m$}}
\put(30,5){\vector(1,0){10}}
\put(30,5){\makebox(10,5){$u_m$}} \put(0,5){\makebox(10,5){$v_m$}}
\put(10,11){\makebox(20,5){$\vdots$}}
\put(0,18){\vector(1,0){10}}
\put(20,18){\oval(20,6)}
\put(10,13){\makebox(20,10){Channel $1$}}
\put(30,18){\vector(1,0){10}}
\put(30,18){\makebox(10,5){$u_1$}} \put(0,18){\makebox(10,5){$v_1$}}
\end{picture}
\caption{A SISO communication scheme.} \label{SISO}
\end{center}
\end{figure}

In this work, we are strongly motivated by the huge MIMO technology recently developed in communication field. A typical MIMO communication system is shown in Fig. \ref{communication}, which is also called a MIMO transceiver. Here the system between $q$ and $p$ is called a MIMO channel characterized by a channel matrix $H$ and an additive white Gaussian noise $d$. The communication engineers are charged to design the transmitter matrix $T$, also referred to as an encoder, and the receiver matrix $R$, also referred to as a decoder, so as to make the received signal $u$ approximate the transmitted signal $v$ as accurately as possible. To fully utilize the advantages of MIMO communication, the transceiver is often built in such a way that the dimensions of $q$ and $p$ are much higher than the dimension of $v$ and $u$. In connection with the NCS as in Fig. \ref{feedback}, we ask out of curiosity the following questions: What will happen if MIMO communication is used in networked control? Does it offer new advantages? Does it lead to new design flexibilities?

\begin{figure}[htbp]
\begin{center}
\begin{picture}(80,15)
\thicklines
\put(0,5){\vector(1,0){10}}
\put(10,0){\framebox(10,10){$T$}}
\put(20,5){\vector(1,0){10}}
\put(30,0){\framebox(10,10){$H$}}
\put(40,5){\vector(1,0){8}}
\put(50,5){\circle{4}}
\put(50,15){\vector(0,-1){8}}
\put(52,5){\vector(1,0){8}}
\put(60,0){\framebox(10,10){$R$}}
\put(70,5){\vector(1,0){10}}
\put(50,5){\makebox(0,0){$+$}}
\put(0,5){\makebox(10,5){$v$}} \put(70,5){\makebox(10,5){$u$}}
\put(20,5){\makebox(10,5){$q$}} \put(52,5){\makebox(8,5){$p$}}
\put(45,7){\makebox(5,8){$d$}}
\end{picture}
\caption{MIMO transceiver, a typical MIMO communication system.} \label{communication}
\end{center}
\vspace{-2pt}
\end{figure}
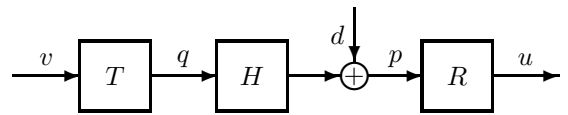

We are also motivated by the following concern. Recall that the channel resource allocation as in \cite{Qiu,Chenb,Xiaob,Xiaoa} is based on a crucial assumption that the channel capacities can be allocated among different input channels subject to a total capacity constraint. What if the individual channel capacities are indeed given a priori and not allocatable?  In that case, is it possible to explore some other design freedoms so as to stabilize the NCSs?

Both motivations lead us to study the following problem. Instead of using a SISO communication scheme, we utilize a MIMO transceiver as in Fig. \ref{communication} to transmit the control signals. For simplicity, we assume that the channel matrix $H$ in the transceiver is identity. In fact, all the developments can be extended straightforwardly to the case of a nonsingular $H$. When $H$ is identity, the MIMO channel in the transceiver becomes a collection of $l$ parallel SISO subchannels. To keep the essence of MIMO technology, we assume that the number of SISO subchannels in the transceiver is greater than or equal to the number of data streams to be transmitted, i.e., $l\geq m$. Later we will see that $l<m$ may also be valid in some cases. This will become more clear as we proceed. The encoder matrix $T\in\mathbb{R}^{l\times m}$ and the decoder matrix $R\in\mathbb{R}^{m\times l}$ are free to be designed subject to certain mild constraint. The specific form of the constraint will depend on the model of the subchannels.
By virtue of the coding mechanism in the MIMO transceiver, each subchannel now transmits a linear combination of the control inputs instead of a single control input as in the conventional SISO communication scheme. In this regard, the subchannels can be considered as helping each other in transmitting the control signals. The current communication system is shown in Fig. \ref{new}.

\begin{figure}[htbp]
\begin{center}
\begin{picture}(90,27)
\thicklines
\put(0,5){\vector(1,0){10}} \put(10,0){\framebox(10,25){$T$}}
\put(20,5){\vector(1,0){10}}
\put(45,5){\oval(30,6)}
\put(34.5,0){\makebox(20,10){Subchannel $l$}}
\put(60,5){\vector(1,0){10}} \put(70,0){\framebox(10,25){$R$}} \put(80,5){\vector(1,0){10}}
\put(80,5){\makebox(10,5){$u_m$}} \put(0,5){\makebox(10,5){$v_m$}}
\put(60,5){\makebox(10,5){$p_l$}} \put(20,5){\makebox(10,5){$q_l$}}
\put(30,11.5){\makebox(20,5){$\vdots$}}
\put(0,20){\vector(1,0){10}} \put(20,20){\vector(1,0){10}}
\put(45,20){\oval(30,6)}
\put(34.5,15){\makebox(20,10){Subchannel $1$}}
\put(60,20){\vector(1,0){10}} \put(80,20){\vector(1,0){10}}
\put(80,20){\makebox(10,5){$u_1$}} \put(0,20){\makebox(10,5){$v_1$}}
\put(60,20){\makebox(10,5){$p_1$}} \put(20,20){\makebox(10,5){$q_1$}}
\put(4.5,15.5){\makebox(0,0){$\vdots$}}
\put(85,15.5){\makebox(0,0){$\vdots$}}
\end{picture}
\caption{A MIMO transceiver as a MIMO communication system in MIMO control.} \label{new}
\end{center}
\vspace{-2pt}
\end{figure}
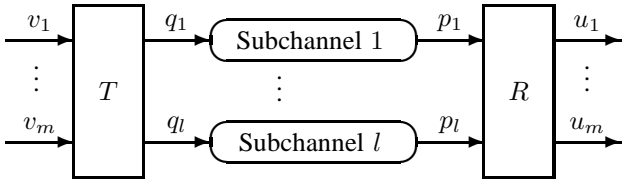

To finish the problem setup, we need to specify the model of the SISO subchannels. As a matter of fact, we consider two different models, each of which addresses a particular type of information constraint imposed on the subchannels.

The first model is concerned with the signal-to-noise ratio constraint. Specifically, all the SISO subchannels in the MIMO transceiver are modeled as AWGN channels as in Fig. \ref{SNR}.
\begin{figure}[htbp]
\begin{center}
\begin{picture}(40,12)
\thicklines
\put(0,0){\vector(1,0){18}}
\put(20,12){\vector(0,-1){10}}
\put(20,0){\circle{4}} \put(22,0){\vector(1,0){18}}
\put(20,0){\makebox(0,0){$+$}}
\put(30,0){\makebox(10,5){$p_i$}}
\put(23,8){\makebox(0,0){$d_i$}}
\put(0,0){\makebox(10,5){$q_i$}}
\end{picture}
\caption{An AWGN subchannel.} \label{SNR}
\end{center}
\vspace{-5pt}
\end{figure}
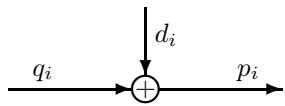
The channel input signal $q_i$ is a stationary process with a certain predetermined admissible power level $P_i$, namely,
\begin{align}
\bfE[q_i^2]<P_i.
\label{powercons}
\end{align}
The signal transmission is disturbed by a zero-mean white Gaussian noise $d_i$ with power spectral density $N_i$.
The ratio $\frac{\bfE[q_i^2]}{N_i}$  is referred to as the signal-to-noise ratio of the channel. Therefore, the predetermined admissible transmission power in turn imposes a signal-to-noise ratio constraint on the channel.
By the knowledge of information theory \cite{Cover}, the capacity of such an AWGN channel with input power constraint $P_i$ is given by
\begin{align}
\mathfrak{C}_i=\frac{1}{2}\frac{P_i}{N_i}. \label{c1}
\end{align}
We define the total channel capacity to be the sum of all the SISO subchannel capacities, i.e.,
\begin{align*}
\mathfrak{C}=\mathfrak{C}_1+\mathfrak{C}_2+\dots+\mathfrak{C}_l.
\end{align*}

It is worth stressing that due to the predetermined admissible power levels $P_i$, the subchannel capacities $\mathfrak{C}_i,i=1,2,\dots,l$, are fixed a priori and thus, cannot be allocated among the subchannels as in \cite{Qiu,Chenb,Xiaob,Xiaoa} any more. Nevertheless, here, the encoder/decoder pair provides a substituted design freedom in place of the channel resource allocation. To be specific, the encoder matrix $T$ and the decoder matrix $R$ are to be designed subject to the following constraint:
\begin{align}
RT=I.\label{AWGNcons}
\end{align}
With this additional design freedom, the controller designer is to jointly design the controller and the encoder/decoder pair so as to stabilize the system subject to the input power constraints (\ref{powercons}). This gives rise to a stabilization problem via coding/control co-design. We are interested in finding a fundamental limitation on the subchannel capacities such that the resulting coding/control co-design problem is solvable.

The second subchannel model is concerned with the fading effect in the signal transmission process. Specifically, all the SISO subchannels in the MIMO transceiver are modeled as fading channels as in Fig. \ref{fading}.
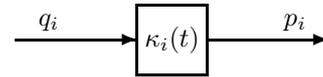
\begin{figure}[htbp]
\begin{center}
\begin{picture}(46,10)
\thicklines
\put(0,5){\vector(1,0){18}}
\put(28,5){\vector(1,0){18}}
\put(18,0){\framebox(10,10){$\kappa_i(t)$}}
\put(36,5){\makebox(10,5){$p_i$}}
\put(0,5){\makebox(10,5){$q_i$}}
\end{picture}
\caption{A fading subchannel.} \label{fading}
\end{center}
\vspace{-5pt}
\end{figure}
Here, the channel input signal $q_i$ is subject to a stochastic multiplicative noise $\kappa_i(t)$ instead of an additive white Gaussian noise, where $\kappa_i(t)$ is a continuous-time white noise process with mean $\mu_i$ and variance $\sigma_i^2$. The mean-square capacity of such a fading channel is defined as
\begin{align}
\mathfrak{C}_i=\frac{1}{2}\frac{\mu_i^2}{\sigma_i^2}.\label{capfading}
\end{align}
The total channel capacity is then given by the sum of all the SISO subchannel capacities, i.e.,
\begin{align*}
\mathfrak{C}=\mathfrak{C}_1+\mathfrak{C}_2+\dots+\mathfrak{C}_l.
\end{align*}
For future use, denote
\begin{align}
M&=\mathrm{diag}\{\mu_1,\mu_2,\dots,\mu_l\},\label{M}\\
\Sigma^2&=\mathrm{diag}\{\sigma_1^2,\sigma_2^2,\dots,\sigma_l^2\}.\label{Sigma2}
\end{align}

We assume that the subchannel capacities $\mathfrak{C}_i,i=1,2,\dots,l$, are given a priori and, thus, not allocatable. As before, the lack of channel resource allocation can be compensated by the substituted design freedom given by the coding mechanism in the MIMO transceiver. A slight difference is that now the encoder/decoder pair is to be designed subject to the following constraint:
\begin{align}
RMT=I,\label{fadingcons}
\end{align}
which simply says that the output signal $p$ has the same mean as the input signal $q$.
Again, the controller designer is to jointly design the controller and the encoder/decoder pair so as to stabilize the system, leading to another stabilization problem via coding/control co-design. We are interested in finding a fundamental limitation on the subchannel capacities such that the resulting coding/control co-design problem is solvable.

Before proceeding, let us recall an important notion called topological entropy. Given a continuous-time linear system $\dot{x}(t)=Ax(t)$, its topological entropy is defined as the quantity
$H(A)\!=\!\sum_{\mathfrak{R}(\lambda_i)>0}\lambda_i$,
where $\lambda_i$ are the eigenvalues of $A$. This quantity has appeared in various studies of optimal control and networked control, for instance, \cite{Baillieul,Braslavsky,Chena,Chenb,Shu,Xiaob}. All these studies point to a very interesting observation: The topological entropy serves as a meaningful instability measure of a linear system based on the minimum resource required to stabilize the system. This observation will be exploited and strengthened in this paper.

\section{Preliminary}
For preparation, some preliminary knowledge is presented in this section, including the cyclic decomposition of a linear system, optimal complementary sensitivity, mean-square norm of a transfer function, as well as some basic concepts and properties in majorization theory.
\vspace{-10pt}
\subsection{Cyclic decomposition}
Let $A$ be an $n\times n$ real matrix. The minimal polynomial of $A$ is the monic polynomial $\alpha(\lambda)$ of least degree such that $\alpha(A)=0$. The minimal polynomial of a matrix is unique. We say that $A$ is cyclic if its minimal polynomial has degree $n$, or equivalently, its minimal polynomial coincides with its characteristic polynomial. In general, given a square matrix, one can always carry out a cyclic decomposition as below.

\begin{lemma}[\cite{Wonham}]
Given a matrix $A\in\mathbb{R}^{n\times n}$ with minimal polynomial $\alpha(\lambda)$, there exists a nonsingular matrix $P$ such that
\begin{align*}
P^{-1}AP=\begin{bmatrix}A_1
    & 0 & \cdots & 0 \\
    0 & A_2 & \ddots & \vdots \\
    \vdots & \ddots & \ddots & 0\\
    0 & \cdots & 0 & A_k\end{bmatrix},
\end{align*}
where $A_i,i=1,2,\dots,k$, are cyclic with minimal polynomials $\alpha_i(\lambda)$, such that $\alpha_1(\lambda)=\alpha(\lambda)$ and $\alpha_{i+1}(\lambda)|\alpha_i(\lambda)$ for $i=1,2,\dots,k-1$.
\label{cycmatrix}
\end{lemma}

\begin{remark}
\label{remarkmatcyc}
Note that in Lemma \ref{cycmatrix}, the number $k$, referred to as the cyclic index, is unique. The minimal polynomials $\alpha_i(\lambda),i=1,2,\dots,k$, are also unique. In addition, from the relation $\alpha_{i+1}(\lambda)|\alpha_i(\lambda)$, it follows  that the spectrum of $A_{i+1}$ is contained in the spectrum of $A_i$. Consequently, there holds $H(A_1)\geq H(A_2)\geq \cdots\geq H(A_k)$.
\end{remark}

The above cyclic decomposition of a matrix further leads to the cyclic decomposition of a linear system which is quite useful in control theory. See the following lemma.

\begin{lemma}[\cite{Wonham}]
\label{lemmasyscyc}
Given a stabilizable linear system $[A|B]$ with $A\in\mathbb{R}^{n\times n}$ and $B\in\mathbb{R}^{n\times m}$, there exist nonsingular matrices $P$ and $Q$ such that
\begin{multline}
[P^{-1}AP|P^{-1}BQ]\\=\left[\!\left.\begin{bmatrix}A_1
    & 0 & \cdots & 0 \\
    0 & A_2 & \ddots & \vdots \\
    \vdots & \ddots & \ddots & 0\\
    0 & \cdots & 0 & A_k\end{bmatrix}\right|
    \!\begin{bmatrix}b_1 & \ast & \cdots & \ast &\ast\\
    0 & b_2 & \ddots & \vdots &\vdots\\
    \vdots & \ddots & \ddots & \ast&\ast\\
    0 & \cdots & 0 & b_k&\ast\end{bmatrix}\!\right],\label{cyclicdec}
\end{multline}
where $A$ is transformed to its cyclic decomposition form and the subsystems $[A_i|b_i],i=1,2,\dots,k$, are stabilizable.
\end{lemma}

\begin{remark}
\label{remarksyscyc}
The subsystems $[A_i|b_i],i=1,2,\dots,k$, are hereinafter referred to as the cyclic subsystems of the system $[A|B]$.
The role of nonsingular matrices $P$ and $Q$ can be considered as linear transformations in the state space and input space, respectively. The following implication can be inferred from Lemma \ref{lemmasyscyc}.
In the cyclic decomposition (\ref{cyclicdec}), $A_1$ contains the greatest number of unstable eigenvalues of $A$ that can be stabilized by a single control input up to linear transformations in the input space; likewise, $A_1$ together with $A_2$ contains the greatest number of unstable eigenvalues of $A$ that can be stabilized by two control inputs up to linear transformations in the input space; and so on so forth.
\end{remark}

\subsection{Optimal complementary sensitivity}
Consider the feedback system with MIMO communication. Assume temporarily that the SISO subchannels in the MIMO transceiver are ideal and the encoder/decoder pair is simply trivial, i.e., $l=m$ and $T=R=I$. Then, the complementary sensitivity function at the plant input is given by
\begin{align*}
\bm{T}(s)=F(sI-A-BF)^{-1}B.
\end{align*}
As shown in many existing works \cite{Braslavsky,Eliab,Shu,Xiaob} as well as later developments in this work, the feedback stabilization in the presence of either additive noise or multiplicative noise is closely related to the $\mathcal{H}_2$ optimal $\bm{T}(s)$.
For preparation, the following lemma gives a solution to $\mathcal{H}_2$ optimal $\bm{T}(s)$.
\begin{lemma}[\cite{Chena}]
\label{ocs}
There holds
\begin{align*}
\inf_{F:A+BF \text{ is stable}}\|\bm{T}(s)\|_2=[2H(A)]^{\frac{1}{2}}.
\end{align*}
Moreover, when $A$ has no eigenvalues on the imaginary axis, the infimum can be achieved by the optimal state feedback gain $F=-B'X$, where $X$ is the stabilizing solution to the algebraic Riccati equation
\begin{align*}
A'X+XA-XBB'X=0.
\end{align*}
\end{lemma}

\subsection{Mean-square norm}
Consider a stable transfer function $\bm{G}(s)$ with dimension $m\times m$. Its mean-square norm is defined to be
\begin{align*}
\|\bm{G}\|_{\mathrm{MS}}=\sqrt{\rho([\|\bm{G}_{ij}\|_2^2])},
\end{align*}
where $\rho(\cdot)$ denotes the spectral radius and $\|\bm{G}_{ij}\|_2$ is the $\mathcal{H}_2$ norm of the $(i,j)$th entry of $\bm{G}(s)$.

By convention, the matrix $1$-norm and $\infty$-norm for $Z\in\mathbb{C}^{p\times m}$ are defined as
\begin{align*}
\|Z\|_1=\max_{1\leq j\leq m}\sum_{i=1}^p |Z_{ij}|,\\
\|Z\|_{\infty}=\max_{1\leq i\leq p}\sum_{j=1}^m |Z_{ij}|,
\end{align*}
respectively. Based on this, we define two mixed norms for $\bm{G}(s)$ as follows:
\begin{align*}
\|\bm{G}\|_{2,1}=\left(\max_{1\leq j\leq m}\sum_{i=1}^m \|\bm{G}_{ij}\|_2^2\right)^{\frac{1}{2}},\\
\|\bm{G}\|_{2,\infty}=\left(\max_{1\leq i\leq m}\sum_{j=1}^m \|\bm{G}_{ij}\|_2^2\right)^{\frac{1}{2}},
\end{align*}
where the subscript $2$ in the mixed norms stands for $\mathcal{H}_2$ norm.

A useful lemma is presented below, which characterizes the connection between the mean-square norm and the two mixed norms. The proof is simply a specialization of Theorem 2 in \cite{Stoer} to the matrix $[\|\bm{G}_{ij}\|_2^2]$ and is thus omitted here for brevity.

\begin{lemma}
There holds
\begin{align*}
\|\bm{G}\|_{\mathrm{MS}}&=\inf_{D\in\mathcal{D}}\|D^{-1}\bm{G}D\|_{2,1}\\
&=\inf_{D\in\mathcal{D}}\|D^{-1}\bm{G}D\|_{2,\infty},
\end{align*}
where $\mathcal{D}$ is the set of all $m\times m$ diagonal matrices with positive diagonal entries.
\end{lemma}

\subsection{Majorization}
As will be seen later, the main result in this paper is given in terms of a majorization type condition. Here, we briefly review some basic concepts and properties in majorization theory. For an extensive treatment of majorization and its applications, one can refer to \cite{MOA}.

For $x, y \in \mathbb{R}^n$, we denote by $x^\downarrow$ and $y^\downarrow$ the rearranged versions of $x$ and $y$ so that their elements are arranged in a non-increasing order. We also denote by $x^\uparrow$ and $y^\uparrow$ the rearranged versions of $x$ and $y$ so that their elements are arranged in a non-decreasing order. We say that $x$ is majorized by $y$, denoted by $x \preccurlyeq y$, if
\begin{align}
x^\downarrow_1 & \leq y^\downarrow_1 \nonumber \\
x^\downarrow_1+x^\downarrow_2 &
\leq y^\downarrow_1+y^\downarrow_2 \nonumber\\
& \vdots \label{bottom} \\
x^\downarrow_1 + x^\downarrow_2 + \dots + x^\downarrow_{n-1} & \leq y^\downarrow_1 + y^\downarrow_2 + \dots + y^\downarrow_{n-1} \nonumber \\
x^\downarrow_1 + x^\downarrow_2 + \dots + x^\downarrow_n & = y^\downarrow_1 + y^\downarrow_2 + \dots + y^\downarrow_n \nonumber
\end{align}
or, equivalently, if
\begin{align}
x^\uparrow_1 & \geq y^\uparrow_1 \nonumber \\
x^\uparrow_1+x^\uparrow_2 &
\geq y^\uparrow_1+y^\uparrow_2 \nonumber \\
& \vdots \label{top} \\
x^\uparrow_1 + x^\uparrow_2 + \dots + x^\uparrow_{n-1} & \geq y^\uparrow_1 + y^\uparrow_2 + \dots + y^\uparrow_{n-1} \nonumber \\
x^\uparrow_1 + x^\uparrow_2 + \dots + x^\uparrow_n & = y^\uparrow_1 + y^\uparrow_2 + \dots + y^\uparrow_n . \nonumber
\end{align}
Quite often, the physical interpretation of majorization is more interesting in applications.
It orders the level of fluctuations
when the averages are the same. In other words, $x\!\preccurlyeq\! y$ says that the elements of $x$ are more even or, less spread out, than the elements of $y$.

Now, if the last equality in (\ref{bottom}) is changed to an inequality $\leq$, $x$ is said to be weakly majorized by $y$ from below,
denoted by $x \preccurlyeq_w y$.
Likewise, if the last equality in (\ref{top}) is changed to an inequality $\geq$, $x$ is said to be weakly majorized by $y$ from above,
denoted by $x \preccurlyeq^w y$.
The weak majorization $\preccurlyeq_w$ and $\preccurlyeq^w$ order the level of
fluctuations and the averages in a combined way.

Further, if all the inequalities $\leq$ in (\ref{bottom}), including the last equality, are changed to strict inequalities $<$, then $x$ is said to be strictly weakly majorized by $y$ from below, denoted by $x \prec_w y$. If all the inequalities $\geq$ in (\ref{top}), including the last equality, are changed to strict inequalities $>$, then $x$ is said to be strictly weakly majorized by $y$ from above, denoted by $x \prec^w y$.

Note that the majorization $\preccurlyeq$ and the weak majorizations $\preccurlyeq_w$ and
$\preccurlyeq^w$ are only pre-orders in $\mathbb{R}^n$. If we say that $x\sim y$ whenever $x=\Pi y$ for some permutation matrix $\Pi$, then $\sim$ defines an equivalence relation on $\mathbb{R}^n$. Denote by $\mathbb{R}^n\!/\!\sim$ the quotient of $\mathbb{R}^n$ with respect to $\sim$. Then the majorization $\preccurlyeq$ and the weak majorizations $\preccurlyeq_w$ and
$\preccurlyeq^w$ are partial orders in $\mathbb{R}^n\!/\!\sim$.

The following lemma characterizes the relation between majorization and weak majorization.
\begin{lemma}[\cite{MOA}]
\label{weaklymaj}
$x \preccurlyeq^w y$ ($x \prec^w y$, respectively), if and only if there exists $z$ such that $x\geq z$ ($x>z$, respectively), and $z\preccurlyeq y$.
\end{lemma}

Another useful lemma is given below.
\begin{lemma}[\cite{MOA}]
\label{consd}There exists a real symmetric matrix $X$ with eigenvalues $\lambda_1,\lambda_2,\dots,\lambda_n$, and diagonal elements $d_1,d_2,\dots,d_n$, if and only if
\begin{align*}
\begin{bmatrix}d_1&d_2&\cdots&d_n\end{bmatrix}'\preccurlyeq\begin{bmatrix}\lambda_1&\lambda_2&\cdots&\lambda_n\end{bmatrix}'.
\end{align*}
\end{lemma}

\vspace{5pt}
When the majorization condition in Lemma \ref{consd} is satisfied, efficient algorithms for finding the desired symmetric matrix $X$ have also been developed in the literature. See for example \cite{Chan,Isaacson}.

\section{Networked Stabilizability via MIMO Communication--AWGN Subchannels}
Starting from this section, we are dedicated to establishing a fundamental limitation on the subchannel capacities required for networked stabilizability via MIMO communication.
To this end, we first consider the case with AWGN subchannels. The case with fading subchannels follows in the next section. In both cases, a unified necessary and sufficient condition is obtained for the stabilizability given in terms of a strictly weak majorization condition.

When the subchannels in the MIMO transceiver are modeled as AWGN channels, the closed-loop system has the form as shown in Fig. \ref{closed}. The encoder matrix $T$ and decoder matrix $R$ are to be designed subject to the constraint (\ref{AWGNcons}).
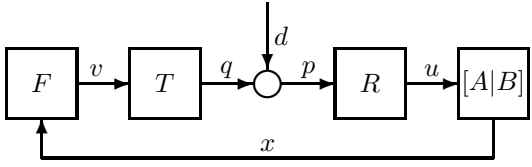
\begin{figure}[htbp]
\begin{center}
\begin{picture}(76,20)
\thicklines
\put(0,5){\framebox(10,10){$F$}}
\put(10,10){\vector(1,0){8}}
\put(18,5){\framebox(10,10){$T$}}
\put(28,10){\vector(1,0){8}}
\put(38,10){\circle{4}}
\put(40,10){\vector(1,0){8}}
\put(48,5){\framebox(10,10){$R$}}
\put(58,10){\vector(1,0){8}}
\put(66,5){\framebox(10,10){$[A|B]$}}
\put(38,22){\vector(0,-1){10}}
\put(71,5){\line(0,-1){6}} \put(71,-1){\line(-1,0){66}}
\put(5,-1){\vector(0,1){6}} \put(32,12){\makebox(0,0){$q$}}
\put(40,17){\makebox(0,0){$d$}} \put(38,1){\makebox(0,0){$x$}}
\put(44,12){\makebox(0,0){$p$}}
\put(13,12){\makebox(0,0){$v$}}
\put(62,12){\makebox(0,0){$u$}}
\end{picture}
\end{center}
\caption{NCS with MIMO communication over AWGN subchannels.}
\label{closed}
\vspace{-5pt}
\end{figure}

Assume that the closed-loop system has reached its steady state and all the signals are wide sense stationary. According to our setup, the total noise $d=\begin{bmatrix}d_1&d_2& \cdots &d_l\end{bmatrix}'$ is a vector white Gaussian noise with power spectral density
\begin{align*}
N=\mathrm{diag}\{N_1,N_2,\dots,N_l\}.
\end{align*}
The complementary sensitivity function, i.e., the closed-loop transfer function from $d$ to $q$, is given by
\begin{align*}
\bm{T}(s)&=TF(sI-A-BRTF)^{-1}BR\\
&=TF(sI-A-BF)^{-1}BR.
\end{align*}
Then, the power spectrum density of $q_i$ has the expression
\[
\{\bm{T}(j\omega)N \bm{T}(j\omega)^*\}_{ii},
\]
and consequently, the power of $q_i$ is given by
\begin{align*}
\bfE[q_i^2]=\frac{1}{2\pi} \int_{-\infty}^{\infty} \{ \bm{T}(j\omega)N \bm{T}(j\omega)^* \}_{ii} d\omega,
\end{align*}
where $\{\cdot\}_{ii}$ stands for the $i$th diagonal element of a matrix.
It follows that the input power constraint (\ref{powercons}) can be rewritten as
\begin{align*}
\frac{1}{2\pi} \int_{-\infty}^{\infty} \{ \bm{T}(j\omega)N \bm{T}(j\omega)^* \}_{ii} d\omega<P_i.
\end{align*}
In view of (\ref{c1}), such constraint can be further translated into
\begin{align}
\frac{1}{2} \frac{1}{2\pi} \int_{-\infty}^{\infty} \left\{N^{-\frac{1}{2}} \bm{T}(j\omega) N \bm{T}(j\omega)^* N^{-\frac{1}{2}} \right\}_{ii}d\omega<\mathfrak{C}_i.\label{constrcap}
\end{align}
The objective is to find requirements on the given subchannel capacities $\mathfrak{C}_i,i=1,2,\dots,l$, such that the networked stabilization can be accomplished subject to the constraints (\ref{constrcap}) via a judicious coding/control co-design. It is also our hope to come up with a systematic procedure for the implementation of the coding/control co-design.

As mentioned before, one of the central issues in solving the coding/control co-design problem is to make the most of the coding mechanism in the MIMO transceiver to facilitate the stabilization of the feedback system. For this purpose, the following understanding provides a significant insight. We are aware that in order to stabilize the NCS, each control input requires certain communication resource for the transmission purpose. As such, the control inputs can be considered as the demand side for the communication resource, while the SISO subchannels in the transceiver are considered as the supply side. The supply capabilities of these SISO subchannels are characterized by their capacities. The challenge lies in the fact that the subchannel capacities are now given a priori and, therefore, the demand and supply may not match in general. To resolve such demand/supply imbalance, it is crucial to observe that the encoder matrix $T$ has an effect of mixing the demands from different control inputs. Bearing this in mind, a smart idea would be to exploit the coding mechanism judiciously so that after the mixing, the demands will be reshaped properly so as to match the supplies. For comparison, the channel resource allocation utilized in \cite{Qiu,Chenb,Xiaob} does the exact opposite, i.e., tailoring the supplies so as to match the demands. It is worth mentioning that demand shaping is a quite general principle in economics. It has led to many successful stories in engineering fields as well such as power systems \cite{Nayyar,Tan}, transportation \cite{Nie}, and data networks \cite{Loiseau}, etc.

The idea of demand shaping turns out to work perfectly. We arrive at a necessary and sufficient condition for the solvability of the coding/control co-design problem given in terms of a strictly weak majorization relation, as shown in the following theorem. The proof is simply making the above understanding precise in a formal mathematical way.

\begin{theorem}
\label{maintheorem}
$[A|B]$ is stabilizable via MIMO communication over AWGN subchannels, if and only if
\begin{multline}
\begin{bmatrix}\mathfrak{C}_1&\mathfrak{C}_2&\cdots&\mathfrak{C}_l\end{bmatrix}'\\\prec^w \begin{bmatrix}H(A_1)&H(A_2)&\cdots&H(A_k)&0&\cdots&0\end{bmatrix}',\label{mr}
\end{multline}
where $H(A_i),i=1,2,\dots,k$, are the topological entropies of the cyclic subsystems $[A_i|b_i]$ as in (\ref{cyclicdec}).
\end{theorem}

\begin{proof}
For brevity, assume that all the eigenvalues of $A$ lie in the open right half complex plane. This assumption can be removed following the same arguments as in \cite{Braslavsky,Chenb,Qiu,Xiaob}.

We first show the necessity. Note that the subchannels can always be reordered so that their capacities are arranged in a non-increasing order. Therefore, without loss of generality, assume that $\mathfrak{C}_1\geq \mathfrak{C}_2\geq \cdots \geq \mathfrak{C}_l$. In view of Remark \ref{remarkmatcyc}, we have $H(A_1)\geq H(A_2)\geq \cdots\geq H(A_k)$.

Suppose there exists a state feedback gain $F$ together with an encoder matrix $T$ and a decoder matrix $R$ such that the NCS is stabilized and the constraints (\ref{constrcap}) are satisfied. It suffices to show
the inequality
\begin{align}
\sum_{i=j}^l\mathfrak{C}_i>\sum_{i=j}^k H(A_i)\label{maine}
\end{align}
holds for $j=1,2,\dots,k$. The case when $j=1$ is quite straightforward since
\begin{align*}
\sum_{i=1}^l\mathfrak{C}_i>&\sum_{i=1}^l\frac{1}{2} \frac{1}{2\pi} \int_{-\infty}^{\infty} \left\{N^{-\frac{1}{2}} \bm{T}(j\omega) N \bm{T}(j\omega)^* N^{-\frac{1}{2}} \right\}_{ii}d\omega\\
=&\frac{1}{2}\|N^{-\frac{1}{2}}\bm{T}(s)N^{\frac{1}{2}}\|_2^2\geq H(A)= \sum_{i=1}^k H(A_i),
\end{align*}
where the second inequality follows from Lemma \ref{ocs}. We proceed to show the case when $j=2$. Let us carry out the controllable-uncontrollable decomposition to the system $[A|B\!R]$ with respect to the first column of $B\!R$, i.e., find a state space transformation $x(t)=Py(t)$ such that the system $[A|B\!R]$ is transformed to
\begin{align*}
\begin{bmatrix}\dot{y}_1(t)\\ \dot{y}_2(t)\end{bmatrix}=\begin{bmatrix}\tilde{A}_{11}&\tilde{A}_{12}\\0&\tilde{A}_{22}\end{bmatrix}\begin{bmatrix}y_1(t)\\ y_2(t)\end{bmatrix}+\begin{bmatrix}\tilde{B}_{11}&\tilde{B}_{12}\\0&\tilde{B}_{22}\end{bmatrix}\begin{bmatrix}p_1(t)\\ \tilde{p}(t)\end{bmatrix},
\end{align*}
where
\begin{align*}
p_1(t)&=q_1(t)+d_1(t),\\
\tilde{p}(t)&=\tilde{q}(t)+\tilde{d}(t),\\
\tilde{q}(t)&=\begin{bmatrix}q_2(t)&q_3(t)&\cdots&q_l(t)\end{bmatrix}',\\
\tilde{d}(t)&=\begin{bmatrix}d_2(t)&d_3(t)&\cdots&d_l(t)\end{bmatrix}'.
\end{align*}
Set $\tilde{F}\!=\!TFP$ and partition $\tilde{F}$ compatibly as $\tilde{F}\!=\!\begin{bmatrix}\tilde{F}_{11}&\tilde{F}_{12}\\ \tilde{F}_{21}&\tilde{F}_{22}\end{bmatrix}$.
Consider the subsystem $[\tilde{A}_{22}|\tilde{B}_{22}]$ with closed-loop dynamics:
\begin{align*}
\dot{y}_2(t)&=\tilde{A}_{22}y_2(t)+\tilde{B}_{22}\tilde{q}(t)+\tilde{B}_{22}\tilde{d}(t),\nonumber\\
\tilde{q}(t)&=\tilde{F}_{22}y_2(t)+\tilde{F}_{21}y_1(t).
\end{align*}
Applying the Laplace transform to both sides of the above equations yields
\begin{align*}
\mathcal{L}(\tilde{q}(t))=\begin{bmatrix}\tilde{\bm{T}}_{21}(s)&\tilde{\bm{T}}_{22}(s)\end{bmatrix}\begin{bmatrix}\mathcal{L}(y_1(t))\\\mathcal{L}(\tilde{d}(t))\end{bmatrix},
\end{align*}
where
\begin{align*}
\tilde{\bm{T}}_{21}(s)&=\tilde{F}_{21}+\tilde{F}_{22}(sI-\tilde{A}_{22}-\tilde{B}_{22}\tilde{F}_{22})^{-1}\tilde{B}_{22}\tilde{F}_{21},\\
\tilde{\bm{T}}_{22}(s)&=\tilde{F}_{22}(sI-\tilde{A}_{22}-\tilde{B}_{22}\tilde{F}_{22})^{-1}\tilde{B}_{22}.
\end{align*}
Since $y_1(t)$ is independent of $\tilde{d}(t)$, we have
\begin{align*}
\bfE[q_{i+1}^2]\geq \frac{1}{2\pi} \int_{-\infty}^{\infty} \{ \tilde{\bm{T}}_{22}(j\omega)\tilde{N} \tilde{\bm{T}}_{22}(j\omega)^* \}_{ii} d\omega,
\end{align*}
for $i=1,2,\dots,l-1$, where $\tilde{N}=\mathrm{diag}\{N_2,N_3,\dots,N_l\}$. Consequently, there holds
\begin{align*}
\sum_{i=2}^{l}\mathfrak{C}_{i}&>\sum_{i=2}^{l}\frac{1}{2}\frac{\bfE[q_{i}^2]}{N_{i}}\\
&\geq \frac{1}{2}\|\tilde{N}^{-\frac{1}{2}}\tilde{\bm{T}}_{22}(s)\tilde{N}^{\frac{1}{2}}\|_2^2\\
&\geq H(\tilde{A}_{22}).
\end{align*}
Meanwhile, it can be inferred from Remark \ref{remarksyscyc} that $H(\tilde{A}_{22})\!\geq\! \sum_{i=2}^k H(A_i)$ and, thus,
$\sum_{i=2}^{l}\mathfrak{C}_{i}\!>\!\sum_{i=2}^k H(A_i)$. In analogy to the above procedure, we can show that the inequality
(\ref{maine}) also holds for $j=3,\dots,k$, which completes the necessity proof.

To show the sufficiency, we will seek a state feedback gain $F$ together with an encoder matrix $T$ and a decoder matrix $R$ such that the NCS is stabilized and the constraints (\ref{constrcap}) are satisfied. Without loss of generality, assume that $[A|B]$ is in the cyclic decomposition form (\ref{cyclicdec}), where each cyclic subsystem $[A_i|b_i],i=1,2,\dots,k$, is stabilizable with state dimension $n_i$. For each $[A_i|b_i]$, we can design a stabilizing state feedback gain $f_i$ such that $\|\bm{T}_i(s)\|_2^2=2H(A_i)$, where
\begin{align}
\bm{T}_i(s)=f_i(sI-A_i-b_if_i)^{-1}b_i.\label{Ti}
\end{align}
The existence of such $f_i$ is guaranteed by Lemma \ref{ocs}.
Let
\begin{align*}
f=\mathrm{diag}\{f_1,f_2,\dots,f_k\},
\end{align*}
and design $F$ to be
\begin{align}
F=\begin{bmatrix}f\\0_{(m\!-\!k)\times n}\end{bmatrix}.\label{F}
\end{align}
It is easy to verify that $F$ is stabilizing, i.e., $A+BF$ is stable. Regarding the encoder/decoder pair, let
\begin{align}
T=N^{\frac{1}{2}}UD^{-1}, \text{ and } R=DU'N^{-\frac{1}{2}},\label{coding}
\end{align}
where $U\in\mathbb{R}^{l\times m}$ is an isometry to be designed and $D=\mathrm{diag}\{1,\epsilon,\dots,\epsilon^{m-1}\}$ with $\epsilon$ being a small positive number.
Also set $S=\mathrm{diag}\{I_{n_1},\epsilon I_{n_2}, \dots, \epsilon^{k-1} I_{n_k}\}$. Then
\begin{align*}
\bm{T}(s)&=TF(sI-A-BF)^{-1}BR\\
&=N^{\frac{1}{2}}U\bar{F}(sI-\bar{A}-\bar{B}\bar{F})^{-1}\bar{B}U'N^{-\frac{1}{2}},
\end{align*}
where
\begin{align}
\bar{F}&=D^{-1}FS=F,\label{barF}\\
\bar{A}&=S^{-1}AS=\begin{bmatrix}
A_1 & 0 & \cdots & 0 \\
0 & A_2 & \ddots & \vdots \\
\vdots & \ddots & \ddots & 0\\
0 & \cdots & 0 & A_k
\end{bmatrix},\label{barA}\\
\bar{B}&=S^{-1}BD=\begin{bmatrix}
b_1 & o(\epsilon) & \cdots & o(\epsilon)&o(\epsilon) \\
0 & b_2 & \ddots & \vdots&\vdots \\
\vdots & \ddots & \ddots & o(\epsilon)&o(\epsilon)\\
0 & \cdots & 0 & b_k&o(\epsilon)
\end{bmatrix},\label{barB}
\end{align}
and $\frac{o(\epsilon)}{\epsilon}$ approaches to a finite constant as $\epsilon\rightarrow 0$. It follows that
\begin{align}
&\frac{1}{2} \frac{1}{2\pi} \int_{-\infty}^{\infty} N^{-\frac{1}{2}} \bm{T}(j\omega) N \bm{T}(j\omega)^* N^{-\frac{1}{2}} d\omega\nonumber\\
&=U\!\left(\!\mathrm{diag}\!\left\{\!\frac{\|\bm{T}_1(s)\|_2^2}{2},\dots,\frac{\|\bm{T}_k(s)\|_2^2}{2},0,\dots,0\!\right\}+o(\epsilon)\!\right)\!U'\nonumber\\
&=U\left(\mathrm{diag}\!\left\{H(A_1),\dots,H(A_k),0,\dots,0\right\}+o(\epsilon)\right)U'.\label{eq1}
\end{align}
When the relation (\ref{mr}) holds, by Lemma \ref{weaklymaj}, there exists a vector $\begin{bmatrix}\gamma_1&\gamma_2&\dots&\gamma_l\end{bmatrix}'$ such that
\begin{align}
\begin{bmatrix}\mathfrak{C}_1&\mathfrak{C}_2&\dots&\mathfrak{C}_l\end{bmatrix}'>\begin{bmatrix}\gamma_1&\gamma_2&\dots&\gamma_l\end{bmatrix}',\label{ine2}
\end{align}
and
\begin{multline}
\begin{bmatrix}\gamma_1&\gamma_2&\cdots&\gamma_l\end{bmatrix}'\\\preccurlyeq \begin{bmatrix}H(A_1)&H(A_2)&\cdots&H(A_k)&0&\cdots&0\end{bmatrix}'.\label{maj1}
\end{multline}
Further, in view of (\ref{maj1}) and Lemma \ref{consd}, an isometry $U$ can be constructed such that
\begin{align}
\left\{U\left(\mathrm{diag}\!\left\{H(A_1),\dots,H(A_k),0,\dots,0\right\}\right)U'\right\}_{ii}=\gamma_i,\label{eq2}
\end{align}
for $i=1,2,\dots,l$. Finally, putting (\ref{eq1}), (\ref{ine2}), and (\ref{eq2}) together, we can verify that the constraints (\ref{constrcap}) are satisfied when $\epsilon$ is sufficiently small. This completes the proof.
\end{proof}

As said before, the majorization $\preccurlyeq$ and weak majorizations
$\preccurlyeq_w$ and $\preccurlyeq^w$ are partial orders in the quotient space $\mathbb{R}^n\!/\!\sim$. For this reason, in the statement of Theorem \ref{maintheorem}, there is no need to assume any monotonicity among the elements of the capacity vector $\begin{bmatrix}\mathfrak{C}_1&\mathfrak{C}_2&\cdots&\mathfrak{C}_l\end{bmatrix}'$. The elements of the right-hand side vector of (\ref{mr}) follow a natural order, i.e., $H(A_1)\geq H(A_2)\geq \dots\geq H(A_k)$. This is due to the property of cyclic decomposition as in Remark \ref{remarkmatcyc}.

The sufficiency proof to Theorem \ref{maintheorem} is constructive. It gives a systematic way to implement the coding/control co-design. The state feedback gain $F$ is designed as in (\ref{F}) by solving $\mathcal{H}_2$ optimal complementary sensitivity for each cyclic subsystem $[A_i|b_i]$. The encoder/decoder pair is designed as in (\ref{coding}), where $D$ is a suitably chosen scaling matrix, and $U$ is an isometry which can be constructed by exploiting the algorithm proposed in \cite{Chan}.

Let us now revisit the intuition of demand shaping to better digest this result. As mentioned before, the topological entropy serves as a measure of the instability of a linear system based on the minimum amount of communication resource required to stabilize the system. From this understanding together with the knowledge of cyclic decomposition, it can be inferred that the vector on the right-hand side of (\ref{mr})
appears as the most uneven demands for the communication resource. Due to the coding mechanism, such demands will be mixed, leading to a set of reshaped demands that are more even, or less spread out. Clearly, the total demand will remain the same after the mixing. Combining these observations yields that in order to enable the reshaped demands to match the given supplies, the supplies must meet the following two requirements:
\begin{enumerate}
\item The total supply should be greater than the total demand, i.e., $\mathfrak{C}>H(A)$.
\item The supplies from different SISO subchannels should be less spread out than the most uneven demands.
\end{enumerate}
It turns out that the above two requirements are exactly what the condition (\ref{mr}) says in light of the physical interpretation of the weak majorization.

What follows is an important implication from \mbox{Theorem \ref{maintheorem}}. Note that we initially assume that the number of subchannels in the MIMO transceiver is greater than or equal to the number of data streams to be transmitted, i.e, $l\geq m$. It turns out that in some cases, it may also be possible to stabilize the NCS with less number of subchannels than the number of data streams. This can be inferred from the majorization type condition (\ref{mr}). In fact, the minimum number of SISO subchannels needed for stabilization is equal to the number of unstable cyclic subsystems $[A_i|b_i]$ yielded from the cyclic decomposition (\ref{cyclicdec}). This observation is in consistence with earlier studies \cite{Heymann,Wonham} in the literature that investigate the minimum number of control inputs required to stabilize a linear system. In that aspect, our result strengthens those studies by indicating a fundamental limitation on the information constraints required for networked stabilization given in terms of a majorization type relation.

One can further deduce the following two corollaries from Theorem \ref{maintheorem}.

\begin{corollary}
\label{cor2}
If the cyclic decomposition of $A$ has only one unstable cyclic block, i.e., $A_1$, then $[A|B]$ is stabilizable via MIMO communication over AWGN subchannels, if and only if $\mathfrak{C}>H(A)$.
\end{corollary}

\begin{proof}
When $A$ has only one unstable cyclic block, the condition (\ref{mr}) reduces to
\begin{align*}
\begin{bmatrix}\mathfrak{C}_1&\mathfrak{C}_2&\cdots&\mathfrak{C}_l\end{bmatrix}'\prec^w\begin{bmatrix}H(A)&0&\cdots&0\end{bmatrix}',
\end{align*}
which holds, if and only if $\mathfrak{C}>H(A)$.
\end{proof}

We wish to mention that Corollary \ref{cor2} is consistent with the result obtained in \cite{Shu}. Also note that Corollary \ref{cor2} includes the single-input system as a special case since a stabilizable single-input system only has one unstable cyclic subsystem. Therefore, this corollary also suggests that in stabilizing a single-input system by using MIMO communication, we only require the total capacity be greater than the topological entropy of the open-loop plant. How the individual subchannel capacities are distributed does not matter in this case.

\begin{corollary}
\label{cor3}
If $\mathfrak{C}_1=\mathfrak{C}_2=\dots=\mathfrak{C}_l$, then $[A|B]$ is stabilizable via MIMO communication over AWGN subchannels, if and only if $\mathfrak{C}>H(A)$.
\end{corollary}

\begin{proof}
When $\mathfrak{C}_1=\mathfrak{C}_2=\dots=\mathfrak{C}_l$, the condition (\ref{mr}) reduces to
\begin{multline}
\frac{1}{l}\begin{bmatrix}\mathfrak{C}&\mathfrak{C}&\cdots&\mathfrak{C}\end{bmatrix}'\\\prec^w \begin{bmatrix}H(A_1)&H(A_2)&\cdots&H(A_k)&0&\cdots&0\end{bmatrix}',\nonumber
\end{multline}
which holds, if and only if $\mathfrak{C}>H(A)$.
\end{proof}

Corollary \ref{cor3} somehow suggests that identical subchannels can best help each other in transmitting the signals.

\section{Networked Stabilizability via MIMO Communication--Fading Subchannels}
The same idea extends to the networked stabilization via MIMO communication over fading subchannels. In this case, the closed-loop system has the form as shown in Fig. \ref{closedfading}, where $\kappa(t)=\mathrm{diag}\{\kappa_1(t),\kappa_2(t),\dots,\kappa_l(t)\}$ stands for the stochastic multiplicative noises in the fading subchannels. The mean and covariance of $\kappa(t)$ are fixed a priori and given by (\ref{M}) and (\ref{Sigma2}), respectively. Therefore, all the subchannel capacities are fixed. As said before, the encoder matrix $T$ and decoder matrix $R$ are to be freely designed subject to a mild constraint as in (\ref{fadingcons}). This enables the controller designer to jointly design the controller and the encoder/decoder pair so as to stabilize the feedback system, leading to a stabilization problem via coding/control co-design.
\begin{figure}[htbp]
\begin{center}
\begin{picture}(84,15)
\thicklines
\put(0,5){\framebox(10,10){$F$}}
\put(10,10){\vector(1,0){10}}
\put(20,5){\framebox(10,10){$T$}}
\put(30,10){\vector(1,0){7}}
\put(37,5){\framebox(10,10){$\kappa(t)$}}
\put(47,10){\vector(1,0){7}}
\put(54,5){\framebox(10,10){$R$}}
\put(64,10){\vector(1,0){10}}
\put(74,5){\framebox(10,10){$[A|B]$}}
\put(79,5){\line(0,-1){6}} \put(79,-1){\line(-1,0){74}}
\put(5,-1){\vector(0,1){6}} \put(33,12){\makebox(0,0){$q$}}
\put(42.5,1){\makebox(0,0){$x$}}
\put(50.5,12){\makebox(0,0){$p$}}
\put(14,12){\makebox(0,0){$v$}}
\put(69,12){\makebox(0,0){$u$}}
\end{picture}
\end{center}
\caption{NCS with MIMO communication over fading subchannels.}
\label{closedfading}
\vspace{-5pt}
\end{figure}
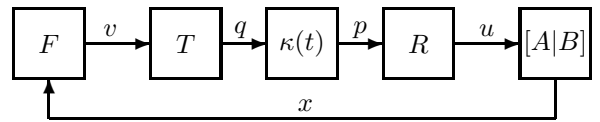

Now, the closed-loop system has the following state space representation:
\begin{align*}
\dot{x}(t)=[A+BR\kappa(t)TF]x(t).
\end{align*}
Denote by $X(t)=\bfE[x(t)x(t)']$ the state covariance. By It\^o's formula \cite{Klebaner}, the evolution of $X(t)$ is governed by the following matrix differential equation:
\begin{multline}
\dot{X}(t)=(A\!+\!BF)X(t)+X(t)(A\!+\!BF)'\\
+BR[\Sigma^2\!\odot\!(TFX(t)F'T')]R'B',\nonumber
\end{multline}
where $\odot$ means Hadamard product. We say that $[A|B]$ is mean-square stabilizable if there exist a state feedback gain $F$ and an encoder/decoder pair such that the closed-loop system in Fig. \ref{closedfading} is mean-square stable, i.e., $\lim_{t\rightarrow\infty}X(t)=0$.

Again, the aim is to find requirements on the subchannel capacities $\mathfrak{C}_i,i=1,2,\dots,l$, such that the networked stabilization can be accomplished via a judicious coding/control co-design.

The next lemma gives several equivalent criteria in verifying the mean-square stabilizability. The equivalence between the implications $(a),(b)$, and $(c)$ can be shown by specializing the well-established results in stochastic control \cite{Boyd,Damm}. The equivalence between the implications $(a)$ and $(d)$ is simply the continuous-time counterpart of Theorem 6.4 in \cite{Eliab}. The details of the proof are omitted here for brevity.
\begin{lemma}
\label{lemmamss}
The following statements are equivalent:
\begin{itemize}
\item[(a)] $[A|B]$ is mean-square stabilizable.
\item[(b)] There exist a feedback gain and an encoder/decoder pair such that the matrix inequality
\begin{multline}
(A\!+\!BF)X+X(A\!+\!BF)'\\
+BR[\Sigma^2\!\odot\!(TFXF'T')]R'B'<0\nonumber
\end{multline}
has a solution $X>0$.
\item[(c)] There exists an encoder/decoder pair such that the matrix inequality
\begin{multline}
A'X+XA\\
\!-\!XBRM[\Sigma^2\!\odot\!(R'B'XBR)]^{-1}\!MR'B'X\!<\!0\label{mrine}
\end{multline}
has a solution $X>0$.
\item[(d)] There exist a feedback gain and an encoder/decoder pair such that
\begin{align*}
\|\bm{T}(s)\Phi\|_{\mathrm{MS}}<1,
\end{align*}
where $\Phi=M^{-1}\Sigma$ and $\bm{T}(s)$ is given by
\begin{align}
\bm{T}(s)=TF(sI-A-BF)^{-1}BRM.\label{cs2}
\end{align}
\end{itemize}
\end{lemma}

The above lemma certainly gives a clue on how to verify the mean-square stabilizability. However, it has not been to our satisfactory yet. Extra effort is needed to reach what we are really looking forward to: an explicit analytic condition on the subchannel capacities such that the NCS can be mean-square stabilized via MIMO communication.

We make our way towards that goal by utilizing exactly the same idea from the previous section. As remarked before, the stabilization of the feedback system requires the balance of the demand and supply of the communication resource. When the subchannel capacities are given a priori, the supply side is not manipulatable. However, by exploiting the coding mechanism in the MIMO transceiver, it is possible to reshape the demands so as to match the supplies. From this perspective, one may expect that the strictly weak majorization condition as in (\ref{mr}) should work equally well in the current case. This is indeed true. A slight difference from the previous case is that now the subchannel capacities take the form as in (\ref{capfading}) \mbox{instead of (\ref{c1}).} The result is summarized in the following theorem.

\begin{theorem}
\label{fadingmr}
$[A|B]$ is mean-square stabilizable via MIMO communication over fading subchannels, if and only if the strictly weak majorization condition (\ref{mr}) holds.
\end{theorem}

\begin{proof}
As before, assume that all the eigenvalues of $A$ lie in the open right half complex plane for brevity.

We first show the necessity. Suppose $[A|B]$ is mean-square stabilizable, we shall show the relation (\ref{mr}) holds. Without loss of generality, assume that $\mathfrak{C}_1\geq \mathfrak{C}_2\geq \dots\geq \mathfrak{C}_l$. It suffices to show the inequality (\ref{maine}) holds for $j=1,2,\dots,k$. Let us start from the case when $j=1$. According to Lemma \ref{lemmamss}, there exists an encoder/decoder pair such that the matrix inequality (\ref{mrine}) has a solution $X>0$. Pre-multiplying and post-multiplying $X^{-\frac{1}{2}}$ on both sides of (\ref{mrine}) yields
\begin{multline}
X^{-\frac{1}{2}}A'X^{\frac{1}{2}}+X^{\frac{1}{2}}AX^{-\frac{1}{2}}\\
-X^{\frac{1}{2}}BRM[\Sigma^2\!\odot\!(R'B'XBR)]^{-1}MR'B'X^{\frac{1}{2}}<0\nonumber.
\end{multline}
Taking trace for both sides of the above inequality yields
\begin{align*}
&\mathrm{tr}(X^{-\frac{1}{2}}A'X^{\frac{1}{2}})+\mathrm{tr}(X^{\frac{1}{2}}AX^{-\frac{1}{2}})\\
&\hspace{20pt}-\mathrm{tr}\big\{X^{\frac{1}{2}}BRM[\Sigma^2\!\odot\!(R'B'XBR)]^{-1}MR'B'X^{\frac{1}{2}}\big\}\\
=&2H(A)-2\sum_{i=1}^l\mathfrak{C}_i<0,
\end{align*}
that validates the inequality (\ref{maine}) when $j=1$. We now proceed to the case when $j=2$. Let us carry out the controllable-uncontrollable decomposition to the linear system $[A|B\!R]$ with respect to the first column of $B\!R$, i.e., find a state space transformation $x(t)=Py(t)$ such that the system $[A|B\!R]$ is transformed to
\begin{align*}
\begin{bmatrix}\dot{y}_1(t)\\ \dot{y}_2(t)\end{bmatrix}=\begin{bmatrix}\tilde{A}_{11}&\tilde{A}_{12}\\0&\tilde{A}_{22}\end{bmatrix}\begin{bmatrix}y_1(t)\\ y_2(t)\end{bmatrix}+\begin{bmatrix}\tilde{B}_{11}&\tilde{B}_{12}\\0&\tilde{B}_{22}\end{bmatrix}\begin{bmatrix}p_1(t)\\ \tilde{p}(t)\end{bmatrix},
\end{align*}
where
\begin{align*}
p_1(t)&=\kappa_1(t)q_1(t),\\
\tilde{p}(t)&=\tilde{\kappa}(t)\tilde{q}(t),\\
\tilde{q}(t)&=\begin{bmatrix}q_2(t)&q_3(t)&\cdots&q_l(t)\end{bmatrix}',\\
\tilde{\kappa}(t)&=\begin{bmatrix}\kappa_2(t)&\kappa_3(t)&\cdots&\kappa_l(t)\end{bmatrix}'.
\end{align*}
Set $\tilde{F}\!=\!TFP$ and partition $\tilde{F}$ compatibly as $\tilde{F}\!=\!\begin{bmatrix}\tilde{F}_{11}&\tilde{F}_{12}\\ \tilde{F}_{21}&\tilde{F}_{22}\end{bmatrix}$.
Consider the subsystem $[\tilde{A}_{22}|\tilde{B}_{22}]$ with closed-loop dynamics
\begin{align*}
\dot{y}_2(t)=[\tilde{A}_{22}+\tilde{B}_{22}\tilde{\kappa}(t)\tilde{F}_{22}]y_2(t)+\tilde{B}_{22}\tilde{\kappa}(t)\tilde{F}_{21}y_1(t),
\end{align*}
where the term $\tilde{F}_{21}y_1(t)$ can be treated as a noise, the power of which goes to zero as time goes to infinity.
Then, by mimicking the above proof in the case when $j=1$, one can easily validate the inequality (\ref{maine}) when $j=2$. Moreover, in the same spirit, one can also validate the inequality (\ref{maine}) when $j=3,\dots,k$, which completes the necessity proof.

For the sufficiency, we will seek a state feedback gain $F$ together with an encoder matrix $T$ and a decoder matrix $R$ such that $\|\bm{T}(s)\Phi\|_{\mathrm{MS}}<1$, where $\bm{T}(s)$ is given by (\ref{cs2}) and $\Phi=M^{-1}\Sigma$. Without loss of generality, assume that $[A|B]$ is in the cyclic decomposition form (\ref{cyclicdec}), where each cyclic subsystem $[A_i|b_i],i=1,2,\dots,k$, is stabilizable with state dimension $n_i$. For each $[A_i|b_i]$, we can design a stabilizing state feedback gain $f_i$ such that $\|\bm{T}_i(s)\|_2^2=2H(A_i)$, where $\bm{T}_i(s)$ is given by (\ref{Ti}).
The existence of such $f_i$ is guaranteed by Lemma \ref{ocs}.
Let
$f=\mathrm{diag}\{f_1,f_2,\dots,f_k\}$
and design $F$ as in (\ref{F}).
Regarding the encoder/decoder pair, let
\begin{align}
T=M^{-\frac{1}{2}}UD^{-1}, \text{ and } R=DU'M^{-\frac{1}{2}},\label{coding2}
\end{align}
where $U\in\mathbb{R}^{l\times m}$ is an isometry to be designed and $D\!=\!\mathrm{diag}\{1,\epsilon,\dots,\epsilon^{m-1}\}$ with $\epsilon$ being a small positive number.
Also set $S=\mathrm{diag}\{I_{n_1},\epsilon I_{n_2}, \dots, \epsilon^{k-1} I_{n_k}\}$. Then
\begin{align*}
\bm{T}(s)&=TF(sI-A-BF)^{-1}BRM\\
&=M^{-\frac{1}{2}}U\bar{F}(sI-\bar{A}-\bar{B}\bar{F})^{-1}\bar{B}U'M^{\frac{1}{2}},
\end{align*}
where $\bar{F}$, $\bar{A}$, and $\bar{B}$ are as in (\ref{barF}), (\ref{barA}), and (\ref{barB}), respectively.
It follows that
\begin{align}
&\|M^{\frac{1}{2}}\bm{T}(s)\Phi M^{-\frac{1}{2}}\|_{2,1}^2\nonumber\\
&\!=\!\max_{1\leq j\leq l}\frac{1}{2\pi}\!\!\int_{-\infty}^{\infty}\!\!\left\{\Phi M^{-\frac{1}{2}} \bm{T}^*(j\omega)M\bm{T}(j\omega)M^{-\frac{1}{2}}\Phi\right\}_{jj}\!\!d\omega\nonumber\\
&\!=\!\max_{1\leq j\leq l}\!\left\{\!\Phi U(\mathrm{diag}\{\|\bm{T}_1(s)\|_2^2,\dots,\|\bm{T}_k(s)\|_2^2,0,\dots,0\}\!\right.\nonumber\\
&\hspace{165pt}\left.+o(\epsilon))U' \Phi\right\}_{jj}\nonumber\\
&\!=\!\max_{1\leq j\leq l} \!\left\{\Phi U(\mathrm{diag}\{2H(A_1),\dots,2H(A_k),0,\dots,0\}\right.\nonumber\\
&\hspace{158pt}\left.+o(\epsilon))U' \Phi\right\}_{jj}.\label{eq11}
\end{align}
When the relation (\ref{mr}) holds, by Lemma \ref{weaklymaj}, there exists a vector $\begin{bmatrix}\gamma_1&\gamma_2&\dots&\gamma_l\end{bmatrix}'$ such that the inequalities (\ref{ine2}) and (\ref{maj1}) are satisfied.
Further, by Lemma \ref{consd}, an isometry $U$ can be constructed such that the equality (\ref{eq2}) holds
for $i=1,2,\dots,l$.
Now, substituting (\ref{eq2}) into (\ref{eq11}) yields
\begin{align*}
\|M^{\frac{1}{2}}\bm{T}(s)\Phi M^{-\frac{1}{2}}\|_{2,1}^2=\max_{1\leq j\leq l}\frac{\gamma_j}{\mathfrak{C}_j}+o(\epsilon).
\end{align*}
It is clear that when $\epsilon$ is sufficiently small, we have
\begin{align*}
\|\bm{T}(s)\Phi\|_{\mathrm{MS}}\leq \|M^{\frac{1}{2}}\bm{T}(s)\Phi M^{-\frac{1}{2}}\|_{2,1}<1,
\end{align*}
which completes the proof.
\end{proof}

\begin{remark}
\label{remarkfading}
From the above proof, one can observe that the controller design is in fact the same as that in the previous case with AWGN subchannels. Specifically, the state feedback gain $F$ is designed as in (\ref{F}) by solving $\mathcal{H}_2$ optimal complementary sensitivity for each cyclic subsystem $[A_i|b_i], i=1,2,\dots,k$. As for the design of the encoder/decoder pair, it looks slightly different from the previous case due to the current constraint (\ref{fadingcons}) as opposed to (\ref{AWGNcons}). However, the essential idea is the same, i.e., reshaping the demands for the communication resource to match the given supplies.
\end{remark}


Notice that all the discussions following Theorem \ref{maintheorem} are also applicable here. Moreover, from Theorem \ref{fadingmr}, one can deduce two corollaries in parallel to Corollary \ref{cor2} and Corollary \ref{cor3}. The proofs are omitted for brevity.

\begin{corollary}
If the cyclic decomposition of $A$ has only one unstable cyclic block, i.e., $A_1$, then $[A|B]$ is MS stabilizable via MIMO communication over fading subchannels, if and only if $\mathfrak{C}>H(A)$.
\end{corollary}

\begin{corollary}
If $\mathfrak{C}_1=\mathfrak{C}_2=\dots=\mathfrak{C}_l$, then $[A|B]$ is mean-square stabilizable via MIMO communication over fading subchannels, if and only if $\mathfrak{C}>H(A)$.
\end{corollary}

Having had all these studies in this section and the previous section, we are now in a position to return to our motivating questions raised in the very beginning. When MIMO communication is used in MIMO control, the encoder/decoder pair in the MIMO transceiver emerges as a substituted design freedom in place of the channel resource allocation. The deployment of MIMO communication brings in many new flexibilities and advantages as summarized below.
\begin{itemize}
\item[1)] The condition on the subchannel capacities for stabilizability is weakened to a great extent. The condition is now given in terms of a strictly weak majorization relation as in (\ref{mr}), which is much weaker than that obtained under the channel resource allocation as in \cite{Qiu,Chenb,Xiaob,Xiaoa}. The understanding obtained in those studies is that not only the total channel capacity should be greater than the topological entropy of the plant, but also the capacities of the individual channels should be greater than certain values.
\item[2)] One essential feature of MIMO communication is that the number of subchannels is often greater than the number of data streams. This has been deployed in wireless communication to increase the data transmission capacity. When applied to networked stabilization, such redundancy in the number of subchannels has the advantage of reducing the capacity requirement in the individual subchannels. For example, with MIMO communication, one can use plenty of subchannels with small capacities to stabilize a single-input system as long as the total capacity is greater than the topological entropy. Note that the idea of introducing redundancy can be found in many other areas as well such as error-correcting codes \cite{Peterson}, over-sampled filter banks \cite{Chai,Cvetkovic}, signal compression \cite{Mallet}, etc.
\item[3)] By virtue of the coding mechanism, in some cases, it may even be possible to stabilize the NCS with less number of subchannels than the number of the control inputs. The minimum number of subchannels needed is equal to the number of unstable cyclic subsystems yielded from the cyclic decomposition.
\end{itemize}

\section{Example}
In this section, a numerical example is provided to illustrate how to carry out the coding/control co-design to stabilize an NCS with MIMO communication.

Consider the following unstable system $[A|B]$:
\begin{align*}
A=\begin{bmatrix}4&0&0&0\\0&2&0&0\\0&0&1&0\\0&0&0&1\end{bmatrix},\quad B=\begin{bmatrix}1&1\\1&1\\1&1\\0&1\end{bmatrix},
\end{align*}
with initial condition $x_0=\begin{bmatrix}1 &1&1&1\end{bmatrix}'$. Clearly, $[A|B]$ is stabilizable. Moreover, it is already in the cyclic decomposition form (\ref{cyclicdec}) with cyclic subsystems
\begin{align*}
[A_1|b_1]=\left[\left.\begin{bmatrix}4&0&0\\0&2&0\\0&0&1\end{bmatrix}\right| \begin{bmatrix}1\\1\\1\end{bmatrix}\right],\text{ and } [A_2|b_2]=[1|1].
\end{align*}
It follows that
$H(A_1)=4+2+1=7$, and $H(A_2)=1$.

We first consider the case when the MIMO transceiver has three AWGN SISO subchannels. For simplicity, let the noise power spectral density be $N=I$. The admissible transmission power levels in the subchannels are given by
\begin{align*}
P_1=9.1, \text{ } P_2=3.1, \text{ and } P_3=4.1.
\end{align*}
In view of (\ref{c1}), the subchannel capacities are
\begin{align*}
\mathfrak{C}_1=4.55, \text{ } \mathfrak{C}_2=1.55, \text{ and }\mathfrak{C}_3=2.05.
\end{align*}
Now, we can verify that the strictly weak majorization relation (\ref{mr}) is satisfied and, thus, the networked stabilization can be accomplished via certain coding/control co-design. One such co-design is carried out as below.

For the controller design, we solve the $\mathcal{H}_2$ optimal $\bm{T}_i(s)$ as in (\ref{Ti}) for each cyclic subsystem $[A_i|b_i],i=1,2$, yielding the optimal feedback gains $f_1=\begin{bmatrix}-40&36&-10\end{bmatrix}$ and $f_2=-2$, respectively. Let
\begin{align}
F=\mathrm{diag}\{f_1,f_2\}=\begin{bmatrix}-40&36&-10&0\\0&0&0&-2\end{bmatrix}.\label{Fsimu}
\end{align}
As for the coding design, let the encoder matrix $T$ and the decoder matrix $R$ be as in (\ref{coding}) with
\begin{align*}
U=\begin{bmatrix}0.7817&0.4714\\0.4629 &0\\-0.4179& 0.8819\end{bmatrix} \text{ and }D=\begin{bmatrix}1&0\\0&0.1\end{bmatrix}.
\end{align*}

With this co-design, we observe that the closed-loop poles are exactly the mirror images of the open-loop poles with respect to the imaginary axis. This validates the stability of the closed-loop system. Moreover, further computation yields
\begin{align*}
\bfE[q_1^2]&=9.0848<P_1,\\
\bfE[q_2^2]&=3.0299<P_2,\\
\bfE[q_3^2]&=4.0249<P_3,
\end{align*}
i.e., the input power constraints (\ref{powercons}) are satisfied. Combining these two observations, we see that the networked stabilization is accomplished via this coding/control co-design.

We proceed to consider the case when the transceiver has three fading SISO subchannels. The mean and covariance of the multiplicative noise $\kappa(t)$ are given by
\begin{align*}
M=\begin{bmatrix}2 &0&0\\0&0.6&0\\0&0&0.9\end{bmatrix}, \;\;\;\Sigma^2=\begin{bmatrix}0.35 &0 &0\\0&0.2&0\\0&0&0.25\end{bmatrix},
\end{align*}
respectively. In view of (\ref{capfading}), the subchannel capacities are
\begin{align*}
\mathfrak{C}_1=5.7143, \text{ } \mathfrak{C}_2=0.9, \text{ and }\mathfrak{C}_3=1.62.
\end{align*}
Once again, we can verify that the strictly weak majorization relation (\ref{mr}) is satisfied and, thus, the mean-square stabilization can be accomplished via certain coding/control co-design. One such co-design is carried out as below.

For the controller design, as in Remark \ref{remarkfading}, we use the same state feedback gain as in (\ref{Fsimu}). For the coding design, let the encoder matrix $T$ and the decoder matrix $R$ be as in (\ref{coding2}) with
\begin{align*}
U=\begin{bmatrix}0.8952&-0.1993\\0&0.8944\\0.4456&0.4004\end{bmatrix} \text{ and }D=\begin{bmatrix}1&0\\0&0.1\end{bmatrix}.
\end{align*}

Under this coding/control co-design, the closed-loop system is mean-square stable. As shown in Fig. \ref{fadingsimu}, the Frobenius norm of the state covariance goes to zero asymptotically.

\begin{figure}[htbp]
\centering
\includegraphics[scale=0.42]{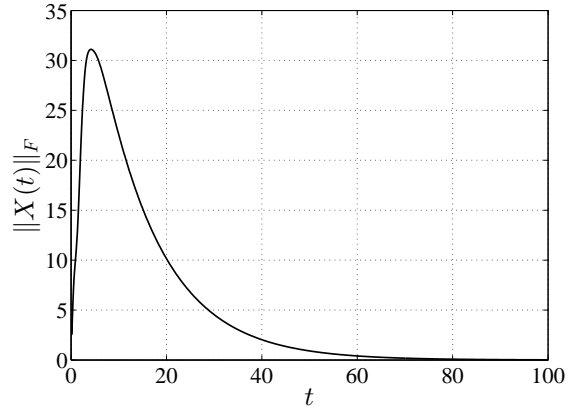}
\caption{Closed-loop evolution of $\|X(t)\|_F$.}
\label{fadingsimu}
\end{figure}

\section{Conclusion}
This paper initiates the study of MIMO control with MIMO communication. We consider the networked stabilization with the communication system between the controller and the plant modeled as a MIMO transceiver, which consists of three parts: an encoder, a MIMO channel, and a decoder. In the spirit of MIMO communication, the number of SISO subchannels in the transceiver is often greater than the number of data streams to be transmitted. The subchannel capacities are fixed a priori and, thus, not allocatable. Nevertheless, the encoder/decoder pair in the transceiver can be designed subject to a certain mild constraint. With this additional design freedom, one needs to design the controller and the encoder/decoder pair jointly so as to stabilize the system, leading to a stabilization problem via coding/control co-design.

The subchannels in the transceiver are modeled in two ways. We first consider the AWGN subchannels and then the fading subchannels.
In both cases, it has been shown that the resulting coding/control co-design problem is solvable, if and only if the majorization type condition (\ref{mr}) is satisfied. The condition (\ref{mr}) relates the subchannel capacities required for stabilization to the topological entropies of the cyclic subsystems of the plant via a majorization type relation. This, on the other hand, gives an application of majorization in control theory. When the relation (\ref{mr}) holds, systematic procedures have also been put forward to carry out the coding/control co-design.

The key for establishing the condition (\ref{mr}) is to observe that the coding mechanism has the effect of mixing the demands for communication resource from different control inputs. As remarked before,
in order to achieve closed-loop stabilization, the demand and supply for the communication resource must be balanced. When the subchannel capacities are fixed a priori, one can design the encoder/decoder pair wisely such that the demands are reshaped properly so as to match the supplies. This contrasts with the channel resource allocation as in \cite{Qiu,Chenb,Feng,Xiaob,Xiaoa} that does the exact opposite, i.e., tailoring the supplies to meet the demands.

Despite the above-mentioned difference, it is worth noting that both co-design approaches provide an additional design freedom on top of the controller design. One common message conveyed is the following: In networked control, it is often more advantageous to design the communication system and the controller jointly rather than separately. By exploiting the additional design freedom, easier problems can be formulated and better results can be obtained. To this extent, the two co-design approaches can be unified under a general framework called communication/control co-design.

As a starting point, this paper is dedicated to the continuous-time networked stabilization with MIMO communication. The discrete-time counterpart is currently under our investigation. In the future, we wish to discover more connections between communication theory and control theory. Also, more applications of majorization in control theory are to be explored.

\section*{Acknowledgment}

The authors wish to thank Professor Weizhou Su of South China University of Technology, Dr. Lidong He of Zhejiang University, and Mr. Wei Jiang of Hong Kong University of Science and Technology for helpful discussions.


\begin{thebibliography}{99}
\bibitem{Baillieul} J.\ Baillieul, ``Feedback coding for information based
control: Operating near the data rate limit'', in {\it Proc. 41th
IEEE Conf. Decision Contr.}, pp. 3229-3236, Las Vegas, NV, 2002.

\bibitem{Boyd} S. Boyd, L. E. Ghaoui, E. Feron, and V. Balakrishnan, \emph{Linear Matrix Inequalities in System and Control Theory}, Society for Industrial and Applied Mathematics, Philadelphia, 1994.

\bibitem{Braslavsky} J.\ H.\ Braslavsky, R.\ H.\ Middleton, and J.\ S.\ Freudenberg,
``Feedback stabilization over signal-to-noise ratio constrained
channels'', {\it IEEE Trans. Automat. Contr.}, vol. 52, pp.
1391-1403, Aug.\ 2007.

\bibitem{Chai} L. Chai, J. Zhang, C. Zhang, and E.  Mosca, ``Frame-theory-based analysis and design of oversampled filter banks: direct computational method'', \emph{IEEE Trans. Signal Process.}, vol. 55, pp. 507-519, Feb. 2007.

\bibitem{Chan} N.\ N.\ Chan and K.-H.\ Li, ``Diagonal elements and eigenvalues of a real symmetric matrix'', \emph{J. Math. Anal. Appl.}, vol. 91, pp. 562-566, 1983.

\bibitem{CY} C.-S.\ Chang and D.\ D.\ Yao, ``Rearrangement, majorization and stochastic scheduling'', {\em Math. Oper. Res.}, vol.\ 18, pp.\ 658-684, 1993.

\bibitem{Chena} W.\ Chen, ``Topological entropy of continuous-time linear systems'', MPhil Thesis, Hong Kong
University of Science and Technology, 2010.

\bibitem{Chenb} W.\ Chen and L.\ Qiu, ``Stabilization of networked control systems with multirate sampling'', \emph{Automatica}, vol. 49, no. 6, pp. 1528-1537, 2013.

\bibitem{CV} F.\ Cicalese and U.\ Vaccaro, ``Supermodularity and subadditivity properties of the entropy on the majorization lattice'', {\em IEEE Trans. Information Theory}, vol.\ 48, pp.\ 933-938, Apr. 2002.

\bibitem{Cover} T. M. Cover and J. A. Thomas, \emph{Elements of Information
Theory}, John Wiley \& Sons, 1991.

\bibitem{Cvetkovic} Z.\ Cvetkovi\'{c} and M.\ Vetterli, ``Oversampled filter banks'', \emph{IEEE Trans. Signal Processing}, vol. 46, pp. 1245-1255, May 1998.

\bibitem{Damm} T.\ Damm, \emph{Rational Matrix Equations in Stochastic Control}, Lecture Notes in Control and Information Sciences, vol.\ 297, Springer-Verlag, Berlin, 2004.

\bibitem{Eliab} N.\ Elia, ``Remote stabilization over fading
channels'', \emph{Systems \& Control Letters}, vol. 54, no. 3, pp. 237-249,
2005.

\bibitem{Eliaa} N.\ Elia and S.\ K.\ Mitter, ``Stabilization of linear
systems with limited information'', {\it IEEE Trans. Automat.
Contr.}, vol. 46, pp. 1384-1400, Sept. 2001.

\bibitem{Feng} Y.\ Feng, X.\ Chen, and G.\ Gu, ``Quantized state feedback control for multiple-input systems subject to signal-to-noise ratio constraints'', in {\it Proc. 52th
IEEE Conf. Decision Contr.}, pp. 7229-7234, Florence, Italy, 2013.

\bibitem{Fu} M.\ Fu and L.\ Xie, ``The sector bound approach to
quantized feedback control'', {\it IEEE Trans. Automat. Contr.},
vol. 50, pp. 1698-1711, Nov. 2005.

\bibitem{Gao} H.\ Gao and T.\ Chen, ``A new approach to quantized feedback control systems'', \emph{Automatica}, vol. 44, no. 2, pp. 534-542, 2008.

\bibitem{Garone} E.\ Garone, B.\ Sinopoli, A.\ Goldsmith, and A.\ Casavola, ``LQG control for MIMO systems over multiple erasure channels with perfect acknowledgment'',
\emph{IEEE Trans. Automat. Contr.}, vol.\ 57, pp.\ 450-456, Feb. 2012.

\bibitem{HLP} G.\ H.\ Hardy, J.\ E.\ Littlewood, and G.\ P\'{o}lya, {\em Inequalities}, 2nd ed., Cambridge University Press, 1952.

\bibitem{Heymann} M.\ Heymann, ``On the input and output reducibility of
multivariable linear systems'', {\it IEEE Trans. Automat. Contr.}, vol. AC-15, pp. 563-569, Oct. 1970.

\bibitem{Isaacson} E.\ L.\ Isaacson, ``Solution to similarity and the diagonal of a matrix'', \emph{Amer. Math. Monthly}, vol. 87, pp. 62-63, 1980.

\bibitem{Klebaner} F. C. Klebaner, \emph{Introduction to Stochastic Calculus with Applications}, 3rd ed., Imperial College Press, London, 2012.

\bibitem{Li} Y. Li, E. Tuncel, J. Chen, and W. Su, ``Optimal tracking performance of discrete-time systems over an additive white noise channel'', in \emph{Proc. 48th IEEE Conf. Decision Contr.}, pp. 2070-2075, Shanghai, P.R.China, 2009.

\bibitem{Lipsa} G. M. Lipsa and N. C. Martins, ``Remote state estimation with communication costs for first-order LTI systems'', {\it IEEE Trans. Automat. Contr.}, vol. 56, pp. 2013-2025, Sept. 2011.

\bibitem{Loiseau} P.\ Loiseau, G.\ Schwartz, J.\ Musacchio, S.\ Amin, and S.\ S.\ Sastry, ``Incentive mechanisms for internet congestion management: fixed-budget rebate versus time-of-day pricing'', \emph{IEEE/ACM Trans. Networking}, to appear, 2014.

\bibitem{Mallet} S.\ Mallat and S.\ Zhong, ``Characterization of signals from multiscale edges'', \emph{IEEE Trans. Pattern Anal. Mach. Intell.}, vol. 14, pp. 710-732, Jul. 1992.

\bibitem{MOA} A.\ W.\ Marshall, I.\ Olkin, and B.\ C.\ Arnold, {\em Inequalities: Theorey of Majorization and Its Applications}, 2nd ed., Springer, New York, 2011.

\bibitem{ne} G.\ N.\ Nair and R.\ J.\ Evans, ``Exponential stabilisability
of finite-dimensional linear systems with limited data rates'', {\it
Automatica}, vol. 39, no. 4, pp. 585-593, 2003.

\bibitem{Nayyar} A.\ Nayyar, M.\ N.-Pincetic, K.\ Poolla and P.\ Varaiya, ``Duration-differentiated energy services'', in preprint, 2014.

\bibitem{Nie} Y.\ Nie and Y.\ Yin, ``Managing rush hour travel choices with tradable credit scheme'', \emph{Transport. Res. B-Meth.}, vol. 50, pp. 1-19, 2013.

\bibitem{PJ} D.\ P.\ Palomar and Y.\ Jiang, {\em MIMO Transceiver Design via Majorization Theory}, Now Publishers Inc., Hanover, MA, USA, 2007.

\bibitem{Peterson} W.\ W.\ Peterson and E.\ J.\ Weldon,\ Jr., \emph{Error-Correcting Codes}, 2nd ed., MIT Press: Cambridge, Mass., 1972.

\bibitem{Qiu} L.\ Qiu, G.\ Gu, and W.\ Chen, ``Stabilization of networked multi-input systems with channel resource allocation'', \emph{IEEE Trans. Automat. Contr.}, vol. 58, pp. 554-568, Mar. 2013.

\bibitem{Shu} Z.\ Shu and R.\ H.\ Middleton, ``Stabilization over power-constrained parallel Gaussian channels'', {\em IEEE Trans.\ Automat.\ Contr.}, vol. 56, pp. 1718-1724, July. 2011.

\bibitem{Stoer} J. Stoer and C. Witzgall, ``Transformations by diagonal matrices in a normed space'',
\emph{Numer. Math.}, vol. 4, pp. 158-171, 1962.

\bibitem{Tan} C.-W.\ Tan and P.\ Varaiya, ``Interruptible electric power service contracts'', \emph{J. Econ. Dynam. Control}, vol. 17, pp. 495-517, May 1993.

\bibitem{Tatikonda} S.\ Tatikonda and S.\ Mitter, ``Control under communication constraints'', {\em IEEE Trans.\ Automat.\ Contr.}, vol. 49, pp. 1056-1068, July. 2004.

\bibitem{TseV} D.\ Tse and P.\ Viswanath, {\em Fundamentals of Wireless Communication}, Cambridge University Press, 2005.

\bibitem{Tsumura} K.\ Tsumura, H.\ Ishii, and H.\ Hoshina, ``Tradeoffs between quantization and packet loss in networked control of linear systems'', \emph{Automatica}, vol. 45, no. 12, pp. 2963-2970, 2009.

\bibitem{Vargas} F.\ J.\ Vargas, E.\ I.\ Silva, and J.\ Chen, ``Stabilization of two-input two-output systems over SNR-constrained channels'', \emph{Automatica}, vol. 49, no. 10, pp. 3133-3140, 2013.

\bibitem{Wonham} W.\ M.\ Wonham, \emph{Linear Multivariable Control: A Geometric \mbox{Approach}}, 3rd ed., Springer-Verlag, 1985.

\bibitem{Xiaob} N. Xiao and L. Xie, ``Feedback stabilization over stochastic multiplicative input channels: continuous-time case'', in \emph{Proc. Int. Conf. Contr., Automation, Robotics and Vision}, pp. 543-548, Singapore, 2010.

\bibitem{Xiaoa} N.\ Xiao, L.\ Xie, and L.\ Qiu, ``Feedback stabilization of discrete-time networked systems over fading channels'',
 \emph{IEEE Trans.\ Automat.\
Contr.}, vol. 57, pp. 2176-2189, Sept. 2012.

\end{thebibliography}
\end{document}